\documentclass[12pt, a4paper]{amsart}
\usepackage[dvips]{epsfig}
\usepackage{amsmath}
\usepackage{amssymb}
\usepackage{amsfonts}
\usepackage{dsfont}
\usepackage{bbm}
\usepackage{amsthm}
\usepackage{amsbsy}
\usepackage{amsgen}
\usepackage{amscd}
\usepackage{tikz}
\usepackage{amsopn}
\usepackage{amstext}
\usepackage{amsxtra}
\usepackage[utf8]{inputenc}
\usepackage{enumerate}
\usepackage{hyperref}
\usepackage{lineno}
\usepackage{marginnote}
\usepackage{verbatim}
\usepackage{calrsfs}
\usepackage{times, graphicx}
\usepackage{layout}
\usepackage{subcaption}
\usepackage{eso-pic}
\usepackage{lastpage}
\DeclareMathAlphabet{\pazocal}{OMS}{zplm}{m}{n}
\usepackage[foot]{amsaddr}

  \evensidemargin
\oddsidemargin
\numberwithin{equation}{section}

\makeatletter
\newcommand*\rel@kern[1]{\kern#1\dimexpr\macc@kerna}
\newcommand*\widebar[1]{%
  \begingroup
  \def\mathaccent##1##2{%
    \rel@kern{0.8}%
    \overline{\rel@kern{-0.8}\macc@nucleus\rel@kern{0.2}}%
    \rel@kern{-0.2}%
  }%
  \macc@depth\@ne
  \let\math@bgroup\@empty \let\math@egroup\macc@set@skewchar
  \mathsurround\z@ \frozen@everymath{\mathgroup\macc@group\relax}%
  \macc@set@skewchar\relax
  \let\mathaccentV\macc@nested@a
  \macc@nested@a\relax111{#1}%
  \endgroup
}
\makeatother

\newtheorem{theorem}{Theorem}[section]
\newtheorem{proposition}[theorem]{Proposition}
\newtheorem{lemma}[theorem]{Lemma}
\newtheorem{corollary}[theorem]{Corollary}

\theoremstyle{definition}

\newtheorem{remark}[theorem]{Remark}

\newcommand{\vol}{\operatorname{Vol}}
\newcommand{\area}{\operatorname{Area}}

\newcommand{\Dc}{\mathcal{D}}
\newcommand{\Dpc}{\pazocal{D}}

\newcommand{\Sc}{\pazocal{S}}

\newcommand{\Ec}{\pazocal{E}}

\newcommand{\Bc}{\mathcal{B}}

\newcommand{\Hc}{\mathcal{H}}

\newcommand {\E} {\mathbb{E}}
\newcommand {\M} {\pazocal{M}}
\newcommand {\R} {\mathbb{R}}

\newcommand {\Z} {\mathbb{Z}}

\newcommand {\Lc} {\mathcal{L}}

\newcommand {\Tb} {\mathbb{T}}
\newcommand {\var} {\operatorname{Var}}

\allowdisplaybreaks
 \allowbreak

\begin{document}

\title[Nodal deficiency]{Nodal deficiency of random spherical harmonics in presence of boundary}

\author{Valentina Cammarota\textsuperscript{1}}
\email{valentina.cammarota@uniroma1.it}
\address{\textsuperscript{1}Department of Statistics, Sapienza University of Rome}
\author{Domenico Marinucci\textsuperscript{2}}
\email{marinucc@mat.uniroma2.it}
\address{\textsuperscript{1}Department of Mathematics, Tor Vergata University of Rome}
\author{Igor Wigman\textsuperscript{3}}
\email{igor.wigman@kcl.ac.uk}
\address{\textsuperscript{3}Department of Mathematics, King's College London}
\dedicatory{Dedicated to the memory of Jean Bourgain}

\date{\today}

\begin{abstract}

We consider a random Gaussian model of Laplace eigenfunctions on the hemisphere satisfying the Dirichlet
boundary conditions along the equator. For this model we find a precise asymptotic law for the corresponding
zero density functions, in both short range (around the boundary) and long range (far away from the boundary)
regimes. As a corollary, we were able to find a logarithmic negative bias for the total nodal length of
this ensemble relatively to the rotation invariant model of random spherical harmonics.

Jean Bourgain's research, and his enthusiastic approach to the nodal geometry of Laplace eigenfunctions, has
made a crucial impact in the field and the current trends within.
His works on the spectral correlations ~\cite[Theorem 2.2]{KKW} and joint with
Bombieri ~\cite{BB} have opened a door for an active ongoing research on the nodal length
of functions defined on surfaces of arithmetic flavour, like the torus or the square.
Further, Bourgain's work ~\cite{Bourgain derandomization}
on toral Laplace eigenfunctions, also appealing to spectral correlations, allowed for inferring deterministic results
from their random Gaussian counterparts.

\end{abstract}
	
\maketitle

\section{Introduction}

\subsection{Nodal length of Laplace eigenfunctions}

The nodal line of a smooth function $f:\M\rightarrow\R$, defined on a smooth compact surface $\M$, with or without a boundary,
is its zero set $f^{-1}(0)$. If $f$ is non-singular, i.e. $f$ has no critical zeros, then its nodal line is a smooth
curve with no self-intersections. An important descriptor of $f$ is its {\em nodal length}, i.e. the length of $f^{-1}(0)$,
receiving much attention in the last couple of decades, in particular, concerning the nodal length of the eigenfunctions
of the Laplacian $\Delta$ on $\M$, in the high energy limit.

Let $(\phi_{j},\lambda_{j})_{j\ge 1}$ be the Laplace eigenfunctions on $\M$, with energies $\lambda_{j}$ in increasing order
counted with multiplicity, i.e.
\begin{equation}
\label{eq:Helmholts eq}
\Delta \phi_{j}+\lambda_{j}\phi_{j}=0,
\end{equation}
endowed with the Dirichlet boundary conditions $\phi|_{\partial \M}\equiv 0$ in presence of nontrivial boundary.
In this context Yau's conjecture asserts that the nodal length
$\Lc(\phi_{j})$ of $\phi_{j}$ is commensurable with $\sqrt{\lambda_{j}}$, in the sense that
$$c_{\M}\cdot \sqrt{\lambda_{j}}\le \Lc(\phi_{j})\le C_{\M}\cdot \sqrt{\lambda_{j}},$$ with some constants $C_{\M}>c_{\M}>0$. Yau's conjecture was resolved for $\M$ analytic ~\cite{Bruning,Bruning-Gromes,DF}, and, more recently, a lower bound ~\cite{Log Lower} and a polynomial upper bound
~\cite{Log Mal,Log Upper} were asserted in {\em full generality} (i.e., for $\M$ smooth).

\subsection{(Boundary-adapted) random wave model}

In his highly influential work ~\cite{Berry 1977} Berry proposed to compare the high-energy Laplace eigenfunctions on generic chaotic surfaces
and their nodal lines to random monochromatic waves and their nodal lines respectively.
The random monochromatic waves (also called Berry's ``Random Wave Model" or RWM) is a
centred isotropic Gaussian random field $u:\R^{2}\rightarrow\R$ prescribed uniquely by the covariance function
\begin{equation}
\label{eq:cov RWM def}
\E[u(x)\cdot u(y)] = J_{0}(\|x-y\|),
\end{equation}
with $x,y\in\R^{2}$ and $J_{0}(\cdot)$ the Bessel $J$ function.

Let
\begin{equation}
\label{eq:K1 density RWM def}
K_{1}^{u}(x)=\phi_{u(x)}(0)\cdot \E[\|\nabla u(x)\|\big| u(x)=0]
\end{equation}
be the zero density, also called the ``first intensity" function of $u$, with $\phi_{u(x)}$ the probability density function of the random variable
$u(x)$.
In this isotropic case, it is easy to directly evaluate
\begin{equation}
\label{eq:K1 dens isotr expl}
K_{1}^{u}(x) \equiv \frac{1}{2\sqrt{2}},
\end{equation}
and then appeal to the Kac-Rice formula, valid under the easily verified non-degeneracy conditions on the random field $u$,
to evaluate the expected nodal length $\Lc(u;R)$ of $u(\cdot)$ restricted to a radius-$R$ disc $\Bc(R)\subseteq\R^{2}$
to be precisely
\begin{equation}
\label{eq:KacRice isotr RWM}
\E[\Lc(u;R)] =\int\limits_{\Bc(R)}K_{1}^{u}(x)dx=
 \frac{1}{2\sqrt{2}}\cdot \area(\Bc(R)).
\end{equation}
Berry ~\cite{Berry 2002}
found that, as $R\rightarrow\infty$, the variance $\var(\Lc(u;R))$ satisfies the asymptotic law
\begin{equation}
\label{eq:log var RWM}
\var(\Lc(u;R)) = \frac{1}{256}\cdot R^{2}\log{R} + O(R^{2}),
\end{equation}
much smaller than the a priori heuristic prediction $\var(\Lc(u;R))\approx R^{3}$ made based on the natural scaling of the problem,
due to what is now known as ``Berry's cancellation" ~\cite{wig} of the leading non-oscillatory term of the $2$-point correlation function (also known as the ``second zero intensity").

\vspace{2mm}

Further, in the same work ~\cite{Berry 2002}, Berry studied the effect induced on the nodal length
of eigenfunctions satisfying the Dirichlet condition on a nontrivial boundary, both in its vicinity and far away from it.
With the (infinite) horizonal axis $\{(x_{1},x_{2}):\: x_{2}=0\}\subseteq \R^{2}$ serving as a model for the boundary, he introduced a Gaussian random field $v(x_{1},x_{2}):\R\times\R_{>0}\rightarrow\R$ of {\em boundary-adapted} (non-stationary) monochromatic random waves,
forced to vanish at $x_{2}=0$. Formally, $v(x_{1},x_{2})$ is the limit, as $J\rightarrow\infty$, of the superposition
\begin{equation*}
\frac{2}{\sqrt{J}}\sum\limits_{j=1}^{J}\sin(x_{2}\sin(\theta_{j}))\cdot \cos(x_{1}\cos(\theta_{j})+\phi_{j})
\end{equation*}
of $J$ plane waves of wavenumber $1$ forced to vanish at $x_{2}=0$.
Alternatively, $v$ is the centred Gaussian random field prescribed by the covariance function
\begin{equation}
\label{eq:cov BARWM def}
r_{v}(x,y):=\E[v(x)\cdot v(y)] = J_{0}(\|x-y\|)-J_{0}(\|x-\widetilde{y}\|),
\end{equation}
$x=(x_{1},x_{2})$, $y=(y_{1},y_{2})$, and $\widetilde{y}=(y_{1},-y_{2})$ is the mirror symmetry of $y$;
the law of $v$ is invariant w.r.t. horizontal shifts
\begin{equation}
\label{eq:v hor shift invar}
v(\cdot,\cdot)\mapsto v(a+\cdot,\cdot),
\end{equation}
$a\in\R$, but not the vertical shifts.

By comparing \eqref{eq:cov RWM def} to \eqref{eq:cov BARWM def}, we observe that,
far away from the boundary (i.e. $x_{2},y_{2}\rightarrow\infty$),
$r_{v}(x,y) \approx J_{0}(\|x-y\|)$, so that, in that range,
the (covariance of) boundary-adapted waves converge to the (covariance of) isotropic ones \eqref{eq:cov RWM def}, though the decay of the error term in this approximation is slow and of oscillatory nature.
Intuitively, it means that, at infinity, the boundary has a small impact on the random waves, though
it takes its toll on the {\em nodal bias}, as it was demonstrated by Berry, as follows.

\vspace{2mm}

Let $K_{1}^{v}(x)=K_{1}^{v}(x_{2})$ be the zero density of $v$, defined analogously to \eqref{eq:K1 density RWM def}, depending on
the height $x_{2}$ only, independent of $x_{1}$ by the inherent invariance \eqref{eq:v hor shift invar}. Berry
showed\footnote{Though a significant proportion of the details of the computation were omitted,
we validated Berry's assertions for ourselves.}
that, as $x_{2}\rightarrow 0$,
\begin{equation}
\label{eq:K1 Berry local}
K_{1}^{v}(x_{2}) \rightarrow \frac{1}{2\pi},
\end{equation}
and attributed this ``nodal deficiency" $\frac{1}{2\pi} < \frac{1}{2\sqrt{2}}$, relatively to \eqref{eq:K1 dens isotr expl},
to the a.s. orthogonality of the nodal lines touching the boundary ~\cite[Theorem 2.5]{Cheng}.

Further, as $x_{2}\rightarrow\infty$,
\begin{equation}
\label{eq:K1 Berry global}
K_{1}^{v}(x_{2}) = \frac{1}{2\sqrt{2}} \cdot \left( 1 + \frac{\cos(2x_{2}-\pi/4)}{\sqrt{\pi x_{2}}} -\frac{1}{32\pi x_{2}}+E(x_{2})  \right) ,
\end{equation}
with some prescribed error term\footnote{Here $E(x_{2})$ is of order $\frac{1}{|x_{2}|}$, so not smaller by magnitude than $\frac{1}{32\pi x_{2}}$, but of oscillatory nature, and will not contribute to the Kac-Rice integral along expanding domains, as neither the term $\frac{\cos(2x_{2}-\pi/4)}{\sqrt{\pi x_{2}}} $.} $E(\cdot)$. In this situation a natural choice for expanding domains are the rectangles $\Dc_{R}:=[-1,1]\times [0,R]$, $R\rightarrow\infty$ (say).
As an application of the Kac-Rice formula \eqref{eq:KacRice isotr RWM} in this case, it easily follows that
\begin{equation}
\label{eq:log neg imp}
\E[\Lc(v;\Dc_{R})]= \frac{1}{2\sqrt{2}}\cdot\area(\Dc_{R}) -\frac{1}{32\sqrt{2}\pi}\log{R} + O(1)
\end{equation}
i.e., a logarithmic ``nodal deficiency" relatively to \eqref{eq:KacRice isotr RWM}, impacted by the boundary infinitely many wave lengths away from it. The logarithmic fluctuations \eqref{eq:log var RWM} in the isotropic case $u$, {\em possibly} also holding for $v$, give rise to a hope to be able to detect the said, also logarithmic, negative boundary impact \eqref{eq:log neg imp} via a single sample of the nodal length, or, at least, very few ones.

\subsection{Random spherical harmonics}

The (unit) sphere $\M=\Sc^{2}$ is one of but few surfaces, where the solutions to the Helmholtz equation \eqref{eq:Helmholts eq}
admit an explicit solution. For a number $\ell\in \Z_{\ge 0}$, the space of solutions of \eqref{eq:Helmholts eq} with
$\lambda=\ell(\ell+1)$ is the $(2\ell+1)$-dimensional space of degree-$\ell$ spherical harmonics, and conversely, all solutions
to \eqref{eq:Helmholts eq} are spherical harmonics of some degree $\ell\ge 0$. Given $\ell\ge 0$, let
$\Ec_{\ell}:=\{\eta_{\ell,1},\ldots \eta_{\ell,2\ell+1}\}$
be any $L^{2}$-orthonormal basis of the space of spherical harmonics of degree $\ell$.
The random field
\begin{equation}
\label{eq:Ttild spher harm full}
\widetilde{T_{\ell}}(x) =
\sqrt{\frac{4 \pi}{2\ell+1}}\sum\limits_{k=1}^{2\ell+1} a_{k}\cdot \eta_{\ell,k}(x),
\end{equation}
with $a_{k}$ i.i.d. standard Gaussian
random variables, is the degree-$\ell$ random spherical harmonics.

The law of $\widetilde{T_{\ell}}$
is invariant w.r.t. the chosen orthonormal basis
$\Ec_{\ell}$, uniquely defined via the covariance function
\begin{equation}
\label{eq:covar RSH}
\E[\widetilde{T_{\ell}}(x)\cdot \widetilde{T_{\ell}}(y)] = P_{\ell}(\cos {d(x,y)}),
\end{equation}
with $P_{\ell}(\cdot)$ the Legendre polynomial of degree $\ell$, and $d(\cdot,\cdot)$ is the spherical distance between $x,y\in\Sc^{2}$.
The random fields $\{\widetilde{T_{\ell}}\}$ are the Fourier components in the $L^{2}$-expansion of {\em every} isotropic random field
~\cite{MP}, of interest, for instance, in cosmology and the study of Cosmic Microwave Background radiation (CMB).

\vspace{2mm}

Let $\Lc(\widetilde{T_{\ell}})$ be the total nodal length of $\widetilde{T_{\ell}}$, of high interest for various pure and applied disciplines, including the above. Berard
~\cite{Berard} evaluated the expected nodal length to be precisely
\begin{equation}
\label{eq:Berard exp}
\E[\Lc(\widetilde{T_{\ell}})] = \sqrt{2\pi}\cdot \sqrt{\ell(\ell+1)},
\end{equation}
and, as $\ell\rightarrow\infty$ its variance is asymptotic ~\cite{wig} to
\begin{equation}
\label{eq:var log spher}
\var(\Lc(\widetilde{T_{\ell}})) \sim \frac{1}{32}\log{\ell},
\end{equation}
in accordance with Berry's \eqref{eq:log var RWM}, save for the scaling, and the invariance of the nodal lines w.r.t. the symmetry $x\mapsto -x$
of the sphere, resulting in a doubled leading constant in \eqref{eq:var log spher} relatively to \eqref{eq:log var RWM} suitably scaled.
A more recent proof ~\cite{MRWHP} of the Central Limit Theorem for $\Lc(\widetilde{T_{\ell}})$, asserting the asymptotic Gaussianity of
$$ \frac{\Lc(\widetilde{T_{\ell}}) - \E[\Lc(\widetilde{T_{\ell}})]}{\sqrt{\frac{1}{32}\log{\ell}}},$$ is sufficiently robust to also yield the Central Limit Theorem, as $R\rightarrow\infty$
for the nodal length $\Lc(u;R)$ of Berry's random waves, as it was recently demonstrated ~\cite{Vidotto}, also claimed by ~\cite{NPR}.

\subsection{Principal results: nodal bias for the hemisphere, at the boundary, and far away}

Our principal results concern the hemisphere $\Hc^{2}\subseteq \Sc^{2}$, endowed with the {\em Dirichlet} boundary conditions along the equator.
We will widely use the spherical coordinates $$\Hc^{2}=\{(\theta,\phi):\: \theta\in [0,\pi/2],\,\phi\in [0,2\pi)\},$$ with the equator identified
with $\{\theta=\pi/2\}\subseteq \Hc^{2}$. Here all the Laplace eigenfunctions are necessarily spherical harmonics restricted to $\Hc^{2}$, subject to some extra properties. Recall that a concrete (complex-valued) orthonormal basis of degree $\ell$ are the Laplace spherical harmonics
$\{Y_{\ell,m}\}_{m=-\ell}^{\ell}$, given in the spherical coordinates by
$$Y_{\ell,m}(\theta,\phi) = e^{im\phi}\cdot P_{\ell}^{m}(\cos{\theta}),$$ with $P_{\ell}^{m}(\cdot)$
the associated Legendre polynomials of degree $\ell$ on order $m$.
For $\ell\ge 0$, $|m|\le \ell$ the spherical harmonic $Y_{\ell,m}$ obeys the Dirichlet boundary condition on the equator, if and only if
$m\not\equiv\ell\mod{2}$, spanning a subspace of dimension $\ell$ inside the $(2\ell+1)$-dimensional space of spherical harmonics of degree $\ell$
~\cite[Example 4]{HT}.
(Its $(\ell+1)$-dimensional orthogonal complement is the subspace satisfying the Neumann boundary condition.) Conversely, every Laplace eigenfunction on $\Hc^{2}$ is necessarily a spherical harmonic of some degree $\ell\ge 0$ that is a linear combination of $Y_{\ell,m}$ with $m\not\equiv\ell\mod{2}$.

\vspace{2mm}

The principal results of this paper concern the following model of {\em boundary-adapted} random spherical
harmonics
\begin{equation}
\label{eq:BARSH def}
T_{\ell}(x) = \sqrt{ \frac{8 \pi }{2\ell+1} }\sum\limits_{\substack{m=-\ell\\m\not\equiv\ell\mod{2}}}^{\ell}a_{\ell,m}Y_{\ell,m}(x),
\end{equation}
where the $a_{\ell,m}$ are the standard (complex-valued) Gaussian random variables subject to the constraint $a_{\ell,-m}=\overline{a_{\ell,m}}$,
so that $T_{\ell}(\cdot)$ is real-valued. Our immediate concern is for the law of $T_{\ell}$, which, as for any centred Gaussian random field, is uniquely determined by its covariance function, claimed by the following proposition.

\begin{proposition} \label{16:52}
The covariance function of $T_{\ell}$ as in \eqref{eq:BARSH def} is given by
\begin{equation}
\label{eq:covar BARSH}
r_{\ell}(x,y):= \E[T_{\ell}(x)\cdot T_{\ell}(y)] = P_{\ell}(\cos d(x,y)) - P_{\ell}(\cos d(x,\overline{y})),
\end{equation}
where $\overline{y}$ is the mirror symmetry of $y$ around the equator, i.e. $y=(\theta,\phi)\mapsto \overline{y} = (\pi-\theta,\phi)$
in the spherical coordinates.
\end{proposition}

It is evident, either from the definition or the covariance, that the law of $T_{\ell}$ is invariant w.r.t. rotations of $\Hc^{2}$ around
the axis orthogonal to the equator, that is, in the spherical coordinates,
\begin{equation}
\label{eq:Tl rot z axis}
T_{\ell}(\theta,\phi)\mapsto T_{\ell}(\theta,\phi+\phi_{0}),
\end{equation}
$\phi\in [0,2\pi)$.
The boundary impact of \eqref{eq:covar BARSH} relatively to \eqref{eq:covar RSH} is in perfect harmony with the boundary impact of the covariance
\eqref{eq:cov BARWM def} of Berry's boundary-adapted model relatively to the isotropic case \eqref{eq:cov RWM def}, except that the mirror symmetry
$y\mapsto \widetilde{y}$ relatively to the $x$ axis in the Euclidean situation is substituted by mirror symmetry $y\mapsto\overline{y}$ relatively to the equator for the spherical geometry. These generalize to $2$ dimensions the boundary impact on the ensemble of stationary random trigonometric polynomials on the circle ~\cite{Qualls,GW} resulting in the ensemble of non-stationary random trigonometric polynomials vanishing at the endpoints ~\cite{Dunnage,ADL}.

Let
\begin{align} \label{K1}
K_{1,\ell}(x) = \frac{1}{\sqrt{2\pi}\cdot \sqrt{\var(T_{\ell}(x))}}\E\big[\|\nabla T_{\ell}(x)\|\big| T_{\ell}(x)=0\big],
\end{align}
be the zero density of $T_{\ell}$, that, unlike the rotation invariant the spherical harmonics \eqref{eq:Ttild spher harm full}, genuinely depends on $x\in\Hc$.
More precisely, by the said invariance w.r.t. \eqref{eq:Tl rot z axis}, the zero density $K_{1,\ell}(x)$ depends on the polar angle $\theta$ only.
We rescale by introducing the variable
\begin{equation}
\label{eq:psi rescale}
\psi = \ell (\pi - 2 \theta),
\end{equation}
and, with a slight abuse of notation, write
$$K_{1,\ell}(\psi)=K_{1,\ell}(x).$$
Our principal result deals with the asymptotics of $K_{1,\ell}(\cdot)$, in two different regimes,
in line with \eqref{eq:K1 Berry local} and \eqref{eq:K1 Berry global} respectively.

\begin{theorem}
\label{thm:main asympt}

\begin{enumerate}

\item For $C>0$ sufficiently large, as $\ell \to \infty$, one has
\begin{align}
\label{eq:nod bias hemi far from boundary}
K_{1,\ell}(\psi)&= \frac{\sqrt{\ell(\ell+1)}}{2 \sqrt 2} \left[  1 +  \sqrt{\frac 2 \pi} \frac{1}{\sqrt \psi} \cos\{(\ell+1/2)\psi/\ell-\pi/4\} -\frac{1}{16 \pi \psi} \right. \\
& \;\; \left.+ \frac{15}{16 \pi \psi}  \cos\{(\ell+1/2)2\psi/\ell-\pi/2\}\right]  +O(\psi^{-3/2} \ell^{-2} ), \nonumber
\end{align}
uniformly for $C < \psi < \pi  \ell$, with the constant involved in the $`O'$-notation absolute.

\vspace{0.5cm}

\item For $\ell\ge 1$ one has the uniform asymptotics
\begin{align}
\label{eq:nod bias hemi close to boundary}
K_{1,\ell}(\psi)
= \frac{\ell}{2 \pi} \left[ 1 + O(\ell^{-1}) + O(\psi^2) \right],
\end{align}
with the constant involved in the $`O'$-notation absolute.
\end{enumerate}

\end{theorem}
Clearly, the statement \eqref{eq:nod bias hemi close to boundary} is asymptotic for $\psi$ small only, otherwise yielding the mere bound $K_{1,\ell}(\psi)=O(\ell)$, which is easy.
As a corollary to Theorem \ref{thm:main asympt}, one may evaluate the asymptotic law of the total expected nodal length of $T_{\ell}$,
and detect the negative logarithmic bias relatively to \eqref{eq:Berard exp}, in full accordance with Berry's \eqref{eq:log neg imp}.

\begin{corollary}  \label{cor}
As $\ell \to \infty$, the expected nodal length has the following asymptotics:
\begin{align*}
\E[\Lc({T_{\ell}})] =   2 \pi   \frac{\sqrt{\ell(\ell+1)}}{2 \sqrt 2} -  \frac{1}{32 \sqrt 2}  \log(\ell) + O(1).
\end{align*}
\end{corollary}

\vspace{2mm}

\subsection*{Acknowledgements}

We are grateful to Ze\'{e}v Rudnick for raising the question addressed within this manuscript.
V.C. has received funding from the Istituto Nazionale di Alta Matematica (INdAM)
through the GNAMPA Research Project 2020 ``Geometria stocastica e campi aleatori".
D.M. is supported by the MIUR Departments of Excellence Program
Math@Tov.

\section{Discussion}

\subsection{Toral eigenfunctions and spectral correlations}

Another surface admitting explicit solutions to the Helmholtz equation \eqref{eq:Helmholts eq} is the standard torus $\Tb^{2}=\R^{2}/\Z^{2}$. Here
the Laplace eigenfunctions with eigenvalue $4\pi^{2}n$ all correspond to an integer $n$ expressible as a sum of two squares, and are given by a sum
\begin{equation}
\label{eq:fn eig torus def}
f_{n}(x)= \sum\limits_{\|\mu\|^{2}=n} a_{\mu}e(\langle \mu,x\rangle)
\end{equation}
over all lattice points $\mu=(\mu_{1},\mu_{2})\in \Z^{2}$ lying on the radius-$\sqrt{n}$ centred circle,
$n$ is a sum of two squares with $e(y):=e^{2\pi iy}$,
$\langle\mu,x\rangle= \mu_{1}x_{1}+\mu_{2}x_{2}$, $x=(x_{1},x_{2})\in\Tb^{2}$. Following ~\cite{ORW}, one endows the eigenspace of $\{f_{n}\}$
with a Gaussian probability measure with the coefficients $a_{\mu}$ standard (complex-valued) i.i.d. Gaussian, save for $a_{-\mu}=\overline{a_{\mu}}$, resulting in the ensemble of ``arithmetic random waves".

\vspace{2mm}

The expected nodal length of $f_{n}$ was computed ~\cite{RW08} to be
\begin{equation}
\label{eq:RW nod length tor}
\E[\Lc(f_{n})] = \sqrt{2}\pi^{2}\cdot \sqrt{n},
\end{equation}
and the useful upper bound
\begin{equation*}
\var(\Lc(f_{n})) \ll \frac{n}{\sqrt{r_{2}(n)}}
\end{equation*}
was also asserted, with $r_{2}(n)$ the number of lattice points lying on the radius-$\sqrt{n}$ circle, or, equivalently, the dimension
of the eigenspace $\{f_{n}\}$ as in \eqref{eq:fn eig torus def}. A precise asymptotic law for $\var(\Lc(f_{n}))$ was subsequently established
~\cite{KKW}, shown to fluctuate, depending on the angular distribution of the lattice points. A non-central non-universal limit
theorem was asserted ~\cite{MRW}, also depending on the angular distribution of the lattice points.

An instrumental key input to both the said asymptotic variance and the limit law
was Bourgain's first nontrivial upper bound ~\cite[Theorem 2.2]{KKW} of $o_{r_{2}(n)\rightarrow\infty}\left(r_{2}(n)^{4}\right)$
for the number of length-$6$ {\em spectral correlations}, i.e. $6$-tuples of lattice points $\{\mu:\:\|\mu\|^{2}=n\}$ summing up to $0$.
Bourgain's bound was subsequently improved and generalized to higher order correlations ~\cite{BB}, in various degrees of generality,
conditionally or unconditionally. These results are still actively used within the subsequent and ongoing research, in particular,
~\cite{Bourgain derandomization} and its followers.

\subsection{Boundary impact}

It makes sense to compare the torus to the square with Dirichlet boundary, and test what kind of impact it would have relatively to
\eqref{eq:RW nod length tor} on the expected nodal length, as the ``boundary-adapted arithmetic random waves", that were addressed in ~\cite{CKW}.
It was concluded, building on Bourgain-Bombieri's ~\cite{BB}, and by appealing to a
different notion of spectral correlation,
namely, the spectral {\em semi-correlations}, that, even at the level of expectation, the total nodal bias is fluctuating from nodal deficiency (negative bias) to nodal surplus (positive bias),
depending on the angular distribution of the lattice points and its interaction with the direction of the square boundary, at least,
for generic energy levels.
A similar experiment conducted by Gnutzmann-Lois for cuboids of arbitrary dimensions, averaging for eigenfunctions admitting separation of variables belonging to different eigenspaces, revealed consistency with Berry's nodal deficiency ansatz
stemming from \eqref{eq:log neg imp}.

\vspace{2mm}

It would be useful to test whether different Gaussian random fields on the square would result in different limiting nodal bias around the boundary corresponding to \eqref{eq:nod bias hemi close to boundary}, that is likely to bring in a different notion of spectral correlation, not unlikely ``quasi-semi-correlation" ~\cite{BMW,KS}. Another question of interest is ``de-randomize" any of these results, i.e. infer the corresponding results on deterministic eigenfunctions following Bourgain ~\cite{Bourgain derandomization}. We leave all of these to be addressed elsewhere.

\section{Joint distribution of $(f_n(x),\nabla f_n(x))$}

In the analysis of $K_{1,\ell}(x)$ we naturally encounter the distribution of $T_{\ell}(x)$, determined by $${\rm Var}(T_{\ell}(x))=1-P_{\ell}(\cos d(x, \bar{x}));$$ and the distribution of $\nabla T_{\ell}(x)$ conditioned on $T_{\ell}(x)=0$, determined by its $2 \times 2$ covariance matrix
$${\bf \Omega}_{\ell}(x)=\mathbb{E}[\nabla T_{\ell}(x) \cdot \nabla^t T_{\ell}(x) | T_{\ell}(x)=0].$$
Let $x$ correspond to the spherical coordinates $(\theta,\phi)$. An explicit computation shows that the covariance matrix ${\bf \Omega}_{\ell}(x)$ depends only on $\theta$, and below we will often abuse notation to write ${\bf \Omega}_{\ell}(\theta)$ instead, and also, when convenient,
${\bf \Omega}_{\ell}(\psi)$ with $\psi$ as in \eqref{eq:psi rescale}. A direct computation shows that:
\begin{lemma} \label{13:48}
The $2 \times 2$ covariance matrix of $\nabla T_{\ell}(x)$ conditioned on $T_{\ell}(x)=0$ is the following real symmetric matrix
\begin{align}
\label{eq:Omega eval}
{\bf \Omega}_\ell(x)= \frac{\ell(\ell+1)}{2} \left[ {\bf I}_2 + {\bf S}_\ell(x) \right],
\end{align}
where
\begin{align*}
{\bf S}_\ell(x)&= \left( \begin{array}{cc}
S_{11,\ell}(x)& 0\\
0& S_{22,\ell}(x)
\end{array}\right),
\end{align*}
and for $x=(\theta,\phi)$
\begin{align*}
S_{11,\ell}(x)&=-\frac{2}{\ell (\ell+1)} \Big[\cos(2 \theta) \; P'_{\ell}(\cos(\pi-2 \theta))+ \sin^2(2 \theta) \; P''_{\ell}(\cos(\pi-2 \theta)) \\
&\;\;+  \frac{1}{1- P_{\ell}(\cos(\pi-2 \theta))}   \sin^2(2 \theta) \; [P'_{\ell}(\cos(\pi-2 \theta)) ]^2\Big],\\
S_{22,\ell}(x)&= -\frac{2}{\ell (\ell+1)}  P'_{\ell}(\cos(\pi-2 \theta)).
\end{align*}
\end{lemma}


In the next two sections we prove Lemma \ref{13:48}, that is, we evaluate the $2 \times 2$ covariance matrix of $\nabla T_{\ell}(x)$ conditioned upon $T_{\ell}(x)=0$. First, in section \ref{uncond}, we evaluate the unconditional $3 \times 3$ covariance matrix ${\bf \Sigma}_{\ell}(x)$ of $(T_{\ell}(x), \nabla T_{\ell}(x))$ and then, in section \ref{con_cov_matrix}, we apply the standard procedure for conditioning multivariate Gaussian random variables.

\subsection{The unconditional covariance matrix} \label{uncond}

The covariance matrix of $$(T_{\ell}(x), \nabla T_{\ell}(x)),$$ which could be expressed as
 \begin{align*}
 {\bf \Sigma}_{\ell}(x)=\left(
\begin{array}{cc}
{\bf A}_{\ell}(x) & {\bf B}_{\ell}(x) \\
{\bf B}_{\ell}^t(x) & {\bf C}_{\ell}(x)
\end{array}
\right),
 \end{align*}
 where
 \begin{align*}
 {\bf A}_{\ell}(x)&={\rm Var}(T_{\ell}(x)),\\
 {\bf B}_{\ell}(x)&=\mathbb{E}[T_{\ell}(x)\cdot  \nabla_{y} T_{\ell}(y)] \big|_{x=y},\\
 {\bf C}_{\ell}(x)&=\mathbb{E}[ \nabla_x T_{\ell}(x) \otimes  \nabla_{y} T_{\ell}(y)] \big|_{x=y}.
 \end{align*}
The $1 \times 2$ matrix ${\bf B}_{\ell}(x)$ is
\begin{align*}
{\bf B}_{\ell}(x)=\left(
\begin{array}{cc}
B_{\ell,1}(x) &
B_{\ell,2}(x)
\end{array}
\right),
 \end{align*}
where ${\bf B}_{\ell}(x)$ depends only on $\theta$, and by an abuse of notation we write
\begin{align*}
B_{\ell,1}(x)&=\frac{\partial}{\partial \theta_{y} }r_\ell(x,y) \Big|_{x=y} = - \sin(2 \theta) \cdot P'_{\ell}(\cos(\pi-2 \theta)), \\
B_{\ell,2}(x)&=\frac{1}{\sin \theta_{y}} \cdot\frac{\partial}{\partial \phi_{y}} r_\ell(x,y) \Big|_{x=y} =0.
\end{align*}
The entries of the $2 \times 2$ matrix ${\bf C}_{\ell}(x)$ are
\begin{align*}
{\bf C}_{\ell}(x)=\left(
\begin{array}{cc}
C_{\ell,11}(x) & C_{\ell,12}(x) \\
C_{\ell,21}(x) & C_{\ell,22}(x)
\end{array}
\right),
 \end{align*}
where again recalling that $x=(\theta,\phi)$ we write
\begin{align*}
C_{\ell,11}(x)&=\frac{\partial}{\partial \theta_x} \frac{\partial}{\partial \theta_{y}}  r_\ell(x,y) \Big|_{x=y}\\
&=P'_{\ell}(1)- \cos(2 \theta) \; P'_{\ell}(\cos(\pi-2 \theta)) -  \sin^2(2 \theta) \; P''_{\ell}(\cos(\pi-2 \theta)),\\
C_{\ell,12}(x)&= C_{\ell,21}(x)= \frac{1}{\sin \theta_{y} }\frac{\partial}{\partial \phi_{y}}   \frac{\partial}{\partial \theta_x}    r_\ell (x,y) \Big|_{x=y} = 0,\\
C_{\ell,22}(x)&= \frac{1}{\sin \theta_{y}}\frac{\partial}{\partial \phi_{y} }  \frac{1}{\sin \theta_x}\frac{\partial}{\partial \phi_x}  r_\ell(x,y) \Big|_{x=y} = P'_{\ell}(1)-P'_{\ell}(\cos(\pi-2 \theta)).
 \end{align*}
 \subsection{Conditional covariance matrix} \label{con_cov_matrix}
The conditional covariance matrix of the Gaussian vector $(\nabla T_{\ell}(x)|T_{\ell}(x)=0)$ is given by the standard Gaussian transition formula:
\begin{align}
\label{eq:Omega transition}
{\bf \Omega}_\ell(x)={\bf C}_\ell(x) -  \frac{1}{\var(T_{\ell}(x))} {\bf B}^t_\ell(x) {\bf B}_\ell(x).
\end{align}
Again taking $x=(\theta,\phi)$ and observing that
\begin{align*}
 \frac{{\bf B}^t_\ell(x) {\bf B}_\ell(x) }{\var(T_{\ell}(x))} =  \frac{1}{1- P_{\ell}(\cos(\pi-2 \theta))}
\left(
\begin{array}{cc}
\sin^2(2 \theta) \cdot [P'_{\ell}(\cos(\pi-2 \theta)) ]^2 &0 \\
0 &0
\end{array}
\right),
\end{align*}
and $$P'_{\ell}(1)=\frac{\ell(\ell+1)}{2},$$ we have
\begin{align*}
{\bf \Omega}_\ell(x) &= \frac{\ell (\ell+1)}{2} {\bf I}_2 \\
&\;\;- \left(
\begin{array}{cc}
  \cos(2 \theta) \cdot P'_{\ell}(\cos(\pi-2 \theta)) + \sin^2(2 \theta)\cdot P''_{\ell}(\cos(\pi-2 \theta)) & 0 \\
0 & P'_{\ell}(\cos(\pi-2 \theta))
\end{array}%
\right) \\
&\;\; -  \frac{1}{1- P_{\ell}(\cos(\pi-2 \theta))}
\left(
\begin{array}{cc}
\sin^2(2 \theta) \cdot [P'_{\ell}(\cos(\pi-2 \theta)) ]^2 &0 \\
0 &0
\end{array}
\right),
\end{align*}
that is the statement of Lemma \ref{13:48}. \\

\vspace{0.5cm}

\section{Proof of Theorem \ref{thm:main asympt}(1): Perturbative analysis away from the boundary}

\subsection{Perturbative analysis}

The asymptotic analysis \eqref{eq:nod bias hemi far from boundary} is in two steps.
First, we evaluate the variance ${\rm Var}(T_{\ell}(x))$ and each entry in ${\bf S}_\ell(x)$ using the high degree asymptotics of the Legendre polynomials and its derivatives (Hilb's asymptotics). In the second step, performed within Proposition \ref{prop:19:31}, we exploit the
analyticity of the
Gaussian expectation \eqref{K1} as a function of the entries of the corresponding {\em non-singular} covariance matrix,
to Taylor expand $K_{1,\ell}(x)$ where both ${\rm Var}(T_{\ell}(x))-1$ and the entries of ${\bf S}_\ell(x)$ are assumed to be small.

\begin{lemma}[Hilb's asymptotics]
\label{hilb0}
\begin{equation*}
P_\ell(\cos \varphi)=\left( \frac{\varphi}{\sin \varphi}\right)^{1/2} J_0((\ell+1/2) \varphi)+\delta_\ell(\varphi),
\end{equation*}
uniformly for $0 \le  \varphi \le \pi-\varepsilon$, where $J_0$ is the Bessel function of the first kind. For the error term we have the bounds
\begin{align*}
\delta_\ell(\varphi) \ll
\begin{cases}
\varphi^2 O(1), & 0 < \varphi \le C/\ell, \\
\varphi^{1/2} O(\ell^{-3/2}), & C/ \ell \le \varphi \le \pi - \varepsilon ,%
\end{cases}%
\end{align*}
where $C$ is a fixed positive constant and the constants involved in the $O$-notation depend on $C$ only.
\end{lemma}

\begin{lemma}
\label{bessel} The following asymptotic representation for the Bessel functions of the first kind holds:
\begin{align*}
J_0(x)&=\left( \frac{2}{\pi x} \right)^{1/2} \cos(x- \pi/4) %
 \sum_{k=0}^\infty (-1)^k g(2k) \; (2 x)^{-2 k}  \\
&\;\;+\left( \frac{2}{\pi x} \right)^{1/2} \cos(x+ \pi/4)
\sum_{k=0}^\infty (-1)^k g(2k+1)\; (2 x)^{-2 k-1},
\end{align*}
where $\varepsilon>0$, $|\arg x|\le \pi-\varepsilon$, $g(0)=1$ and $g(k)=\frac{(-1)(-3^2) \cdots (-(2k-1)^2)}{2^{2k} k!}=(-1)^k \frac{[(2k)!!]^2}{2^{2k} k!}$.
\end{lemma}

For a proof of Lemma \ref{hilb0} and Lemma \ref{bessel} we refer to \cite[Theorem 8.21.6]{szego} and \cite[section 5.11]{lebedev} respectively.  \\

Recall the scaled variable $\psi$ related to $\theta$ via \eqref{eq:psi rescale}, so that an application of lemmas \ref{hilb0} and \ref{bessel}, yields that, for $\ell \ge 1$ and $C< \psi < \ell \pi$,
\begin{align*}
P_{\ell}(\cos(\psi/\ell))
&=\sqrt{\frac{2}{\pi}} \frac{\ell^{-1/2}}{\sin^{1/2} (\psi/\ell)}\Big[ \cos ((\ell+1/2)\psi/\ell-\pi/4) -   \frac{1}{2 \ell \psi/\ell}   \cos ((\ell+1/2)\psi/\ell+\pi/4)   \Big]\\
&\;\;+  O((\psi/\ell)^{1/2} \ell^{-3/2} ).
\end{align*}
Observing that
\begin{align*}
&\frac{\ell^{-1/2}}{\sin^{1/2} (\psi/\ell)} = \ell^{-1/2} \left[\frac{1}{\sqrt{\psi/\ell}} +O((\psi/\ell)^{\frac 3 2})  \right] = \frac{1}{\sqrt \psi} + O(\psi^{3/2} \ell^{-2}),\\
&\frac{\ell^{-1/2}}{\sin^{1/2} (\psi/\ell)}   \frac{1}{2 \psi} =O( \psi^{-3/2}),
\end{align*}
we write
\begin{align} \label{22:31}
P_{\ell}(\cos(\psi/\ell))=\sqrt{\frac{2}{\pi}}\frac{1}{\sqrt \psi} \cos ((\ell+1/2)\psi/\ell-\pi/4)+O(\psi^{-3/2})+O(\psi^{3/2} \ell^{-2}).
\end{align}
A repeated application of lemmas \ref{hilb0} and \ref{bessel} also yields an asymptotic estimate for the first couple of derivatives of the Legendre Polynomials \cite[Lemma 9.3]{CMW}:
\begin{align} 
P'_{\ell }(\cos (\psi/\ell) )&=\sqrt{\frac{2}{\pi }}\frac{\ell^{1-1/2}}
{\sin ^{1+1/2} (\psi/\ell) }\left[\sin ((\ell+1/2) \psi/\ell-\pi/4) -\frac{1}{8\ell \psi/\ell }\sin
((\ell+1/2)\psi/\ell+\pi/4) \right] \nonumber \\
&\;\;+O(\ell^{- \frac 1 2} (\psi/\ell)^{-\frac 5 2 }), \nonumber
\end{align}
and
\begin{align}  
&P''_{\ell }(\cos (\psi/\ell) ) \nonumber\\
&=\sqrt{\frac{2}{\pi }} \frac{\ell^{2-1/2}}{\sin ^{2+1/2} (\psi/\ell) }\left[-\cos ((\ell+1/2)\psi/\ell-\pi/4) +\frac{1}{8\ell \psi/\ell }\cos ((\ell+1/2)\psi/\ell+\pi/4)\right] \nonumber \\
&\;\; -\sqrt{\frac{2}{\pi }} \frac{\ell ^{1-1/2}}{\sin ^{3+1/2} (\psi/\ell) }\left[\cos
((\ell-1+1/2)\psi/\ell+\pi/4)+\frac{1}{8\ell \psi/\ell }\cos ((\ell-1+1/2)\psi/\ell-\pi/4) \right] \nonumber \\
&\;\; +O(\psi^{-7/2} \ell^{4}). \nonumber
\end{align}

\vspace{2mm}

Since we have that
\begin{align*}
&\frac{\ell^{1-1/2}}{\sin^{1+1/2} (\psi/\ell) }=\ell^{1-1/2} \left[ \frac{1}{(\psi/\ell)^{3/2}}+O((\psi/\ell)^{1/2}) \right]= \frac{\ell^{2}}{\psi^{3/2}} +O(\psi^{1/2}) \\
&\frac{\ell^{1-1/2}}{\sin ^{1+1/2} (\psi/\ell) }  \frac{1}{ \psi} =O( \psi^{-5/2}  \ell^2),
\end{align*}
we have
\begin{align} \label{22:32}
P'_{\ell }(\cos (\psi/\ell) )=\sqrt{\frac{2}{\pi }}    \frac{\ell ^{1-1/2}}{\sin ^{1+1/2} (\psi/\ell) } \sin ((\ell+1/2)\psi/\ell-\pi/4) + O(\psi^{-5/2} \ell^{2}), \end{align}
and observing that
\begin{align*}
& \frac{\ell^{2-1/2}}{\sin ^{2+1/2} (\psi/\ell) }=\ell^{2-1/2} \left[\frac{1}{(\psi/\ell)^{5/2}} + O((\psi/\ell)^{- 1/2} )\right]=\frac{\ell^4}{\psi^{5/2}} +(\psi^{-1/2} \ell^2) \\
& \frac{\ell^{2-1/2}}{\sin ^{2+1/2} (\psi/\ell) } \frac{1}{\psi} =O(\psi^{-7/2} \ell^{4}) \\
&  \frac{\ell^{1-1/2}}{\sin ^{3+1/2} (\psi/\ell) }=\ell^{1-1/2}  \left[\frac{1}{(\psi/\ell)^{7/2}} + O((\psi/\ell)^{-3/2})  \right]=\frac{\ell^4}{\psi^{7/2}}+O(\psi^{-3/2}  \ell^2)
\end{align*}
we obtain
\begin{equation} \label{22:33}
\begin{split}
P''_{\ell }(\cos (\psi/\ell) ) &=- \sqrt{\frac{2}{\pi }}  \frac{\ell^{2-1/2}}{\sin ^{2+1/2} (\psi/\ell) } \cos ((\ell+1/2)\psi/\ell-\pi/4)  \\&+O(\psi^{-3/2} \ell^2) +O(\psi^{-7/2} \ell^4). 
\end{split}
\end{equation}
The estimates in \eqref{22:31}, \eqref{22:32} and \eqref{22:33}, imply that  for $\ell \ge 1$ and uniformly for $C< \psi < \ell \pi$, with $C>0$, we have
\begin{align} \label{sell}
P_{\ell}(\cos(\psi/\ell))&=\sqrt{\frac{2}{\pi}}\frac{1}{\sqrt \psi}\cos ((\ell+1/2)\psi/\ell -\pi/4)+O(\psi^{-3/2})+O(\psi^{3/2} \ell^{-2}),\\
\{P_{\ell}(\cos(\psi/\ell))\}^2&= \frac{2}{\pi} \frac{1}{ \psi} \cos^2 ((\ell+1/2)\psi/\ell -\pi/4)+O(\psi^{-2})+O(\psi \ell^{-2}).  \nonumber
\end{align}
With the same abuse of notation as above,
we write ${\bf S}_\ell(\psi):={\bf S}_\ell(x)$ as in Lemma \ref{13:48}, and in analogous manner for its individual entries
$S_{ij;\ell}(\psi):=S_{11;\ell}(x)$. We have
\begin{align}
S_{11;\ell}(\psi)&=2 \sqrt{\frac 2 \pi} \frac{1}{\sqrt \psi} \cos((\ell+1/2)\psi/\ell-\pi/4)- \frac{4}{\pi} \frac{1}{\psi}  \sin^2((\ell+1/2)\psi/\ell-\pi/4)  \label{Sell11} \\&
\;\;+O(\psi^{-3/2})+O(\psi^{3/2} \ell^{-2}), \nonumber \\
S_{22;\ell}(\psi)
&=- 2 \sqrt{\frac{2}{\pi }}  \frac{1}{\psi^{3/2}} \sin ((\ell+1/2)\psi/\ell-\pi/4) +O(\psi^{1/2} \ell^{-2})+ O(\psi^{-5/2} ). \label{Sell22}
\end{align}

The next proposition prescribes a precise asymptotic expression for the density function $K_{1,\ell}(\cdot)$ via a Taylor expansion of the relevant Gaussian expectation as a function of the associated covariance matrix entries.

\begin{proposition} \label{prop:19:31}
For $C >0$ sufficiently large we have the following expansion on $C < \psi < \ell \pi$:
\begin{align} \label{taylorx}
K_{1,\ell}(\psi) &=\frac{ \sqrt{\ell(\ell+1)}}{2 \sqrt 2} + L_{\ell}(\psi)+E_{\ell}(\psi),
\end{align}
with the leading term
\begin{align*}
L_{\ell}(\psi)&=\frac{ \sqrt{\ell(\ell+1)}}{4 \sqrt 2} \left[  s_{\ell}(\psi) + \frac{1}{2 } {\rm tr} \,{\bf S}_{\ell}(\psi) + \frac{3}{4 } s^2_{\ell}(\psi) + \frac{1}{4 } s_{\ell}(\psi)\,{\rm tr}\,{\bf S}_{\ell}(\psi)- \frac{1}{16 } {\rm tr} \, {\bf S}^2_{\ell}(\psi) - \frac{1}{32 }    ({\rm tr} \, {\bf S}_{\ell}(\psi))^2   \right],
\end{align*}
where $s_{\ell}(\psi)=P_{\ell}(\cos(\psi/\ell))$, and the error term $E_{\ell}(\psi)$ is bounded by
\begin{align*}
|E_{\ell}(\psi)|=O(\ell \cdot (|s_\ell(\psi)|^3 + |{\bf S}_\ell(\psi)|^3 )),
\end{align*}
with constant involved in the $O$-notation absolute.
\end{proposition}

\begin{proof}
To prove Proposition \ref{prop:19:31} we perform a precise Taylor analysis for the density function $K_{1,\ell}(\psi)$, assuming that both $s_{\ell}(\psi)$ and the entries of ${\bf S}_{\ell}(\psi)$ are small.
We introduce the scaled covariance matrix (see \eqref{eq:Omega eval})
\begin{align*}
{\bf \Delta}_\ell(\psi)= \frac{2}{\ell (\ell+1)} {\bf \Omega}_\ell(\psi) &= {\bf I}_2 + {\bf S}_\ell(\psi).
\end{align*}
The density function $K_{1,\ell}(\cdot)$ could be expressed as
\begin{align*}
K_{1,\ell}(\psi)= \frac{1}{\sqrt{2 \pi}} \frac{1}{\sqrt{1-s_\ell(\psi)}}  \frac{1}{2 \pi  \sqrt{  \text{det} \, {\bf \Delta}_\ell(\psi) }}
\frac{\sqrt{\ell(\ell+1)}}{\sqrt 2} \iint_{\mathbb{R}^2} ||z|| \exp \Big\{ -\frac 1 2 z {\bf \Delta}^{-1}_\ell(\psi) z^t \Big\} d z,
\end{align*}
On $(C,\pi \ell)$, with $C$ sufficiently large, we Taylor expand
\begin{align*}
\frac{1}{\sqrt{1-s_\ell(\psi)}}=1+ \frac{1}{2} s_\ell(\psi) + \frac 3 8 s^2_\ell(\psi) +O(s^3_\ell(\psi)),
\end{align*}
since, using the high degree asymptotics of the Legendre polynomials (Hilb's asymptotics), we see that $|P_{\ell}(\cos(\psi/\ell))|$ is bounded away from $1$.  Next, we consider the Gaussian integral
\begin{align*}
\mathcal{I}({\bf S}_\ell(\psi))=  \iint_{\mathbb{R}^2} ||z||  \exp \Big\{ -\frac 1 2 z (I_2 + {\bf S}_\ell(\psi))^{-1} z^t \Big\} d z,
\end{align*}
observing that on $(C,\pi \ell)$, for $C$ sufficiently large, we can Taylor expand
\begin{align*}
 (I_2 + {\bf S}_\ell(\psi))^{-1}&=I_2 - {\bf S}_\ell(\psi) + {\bf S}^2_\ell(\psi) +O({\bf S}^3_\ell(\psi)),
\end{align*}
and the exponential as follows
\begin{align*}
&\exp \Big\{ -\frac 1 2 z (I_2 + {\bf S}_\ell(\psi))^{-1} z^t \Big\}\\
&\;\; = \exp \Big\{ -\frac {z z^t} 2  \Big\}    \Big[ 1+ \frac 1 2 z \Big( {\bf S}_\ell(\psi) - {\bf S}^2_\ell(\psi) +O({\bf S}^3_\ell(\psi)) \Big)  z^t \\
&\hspace{0.5cm} + \frac{1}{2} \Big( \frac{1}{2} z ( {\bf S}_\ell(\psi) - {\bf S}^2_\ell(\psi) +O({\bf S}^3_\ell(\psi)))  z^t  \Big)^2 + O\Big(  z ( {\bf S}_\ell(\psi) - {\bf S}^2_\ell(\psi) +O({\bf S}^3_\ell(\psi)))  z^t  \Big)^3  \Big],
\end{align*}
so that
\begin{align*}
\mathcal{I}({\bf S}_\ell(\psi))
&=  \iint_{\mathbb{R}^2} ||z|| \exp \Big\{ -\frac {z z^t} 2  \Big\}    \Big[ 1+ \frac 1 2 z {\bf S}_\ell(\psi)  z^t  -  \frac 1 2 z {\bf S}^2_\ell(\psi)  z^t + \frac{1}{8} \Big(  z  {\bf S}_\ell(\psi)   z^t  \Big)^2    \Big] d z+ O({\bf S}^3_\ell(\psi)).
\end{align*}

We introduce the following notation:
\begin{align*}
\mathcal{I}_0({\bf S}_\ell(\psi))&=  \iint_{\mathbb{R}^2} ||z|| \exp \Big\{ -\frac { z z^t} 2 \Big\}  d z= 2 \pi \int_{0}^{\infty} \rho \exp\big\{-\frac{1}{2}  \rho^2 \big \} \rho  \, d \rho = 2 \pi \sqrt{\frac{\pi}{2}}  = \sqrt{2} \pi^{3/2},
\end{align*}
\begin{align*}
\mathcal{I}_1({\bf S}_\ell(\psi))&= \frac 1 2  \iint_{\mathbb{R}^2} ||z|| \exp \Big\{ -\frac {z z^t} 2  \Big\}  z {\bf S}_\ell(\psi)  z^t d z
=  \frac{3 }{2^{3/2}} \pi^{3/2} {\rm tr} \, {\bf S}_{\ell}(\psi),
\end{align*}
and
\begin{align*}
\mathcal{I}_2({\bf S}_\ell(\psi))&= -  \frac 1 2  \iint_{\mathbb{R}^2} ||z|| \exp \Big\{ -\frac {z z^t} 2  \Big\}      z {\bf S}^2_\ell(\psi)  z^t  d z \\
&= -  \frac 1 2  \iint_{\mathbb{R}^2} ||z|| \exp \Big\{ -\frac {z z^t} 2  \Big\}     \big( S^2_{11;\ell}(\psi) z_1^2 +S^2_{22;\ell}(\psi) z_2^2 \big)  d z   \\
&= - \frac{3 }{2^{3/2}} \pi^{3/2} \; {\rm tr} \, {\bf S}_{\ell}(\psi).
\end{align*}
We also define
\begin{align} \label{14:22}
\mathcal{I}_3({\bf S}_\ell(\psi))&= \frac{1}{8}  \iint_{\mathbb{R}^2} ||z|| \exp \Big\{ -\frac { z z^t} 2 \Big\}   \Big(  z  {\bf S}_\ell(\psi)   z^t  \Big)^2 d z \nonumber \\
&=\frac{1}{8}  \iint_{\mathbb{R}^2} ||z|| \exp \Big\{ -\frac {z z^t} 2  \Big\}   \big(  S^2_{11;\ell}(\psi) z_1^4+S^2_{22;\ell}(\psi) z_2^4
+2S_{11;\ell}(\psi) S_{22;\ell}(\psi) z_1^2 z_2^2 \big) d z,
\end{align}
and note that
\begin{equation}
\begin{split}
\label{eq:(z1^2+z2^2)^2}
&\iint_{\mathbb{R}^2} ||z|| \exp \Big\{ -\frac {z z^t} 2  \Big\} (z_1^2 + z_2^2)^2 d z = 2 \pi \int_0^{\infty} \rho \exp \Big\{ - \frac{\rho^2}{2} \Big\} \rho^4 \rho d \rho = 2 \frac{15}{\sqrt 2} \pi^{3/2},\\
 &\iint_{\mathbb{R}^2} ||z|| \exp \Big\{ -\frac {z z^t} 2  \Big\}   z_1^4 d z =  \frac{15}{ \sqrt 2}  \frac{3}{4} \pi^{3/2},
\end{split}
\end{equation}
and that
\begin{equation}
\label{int z1^2z2^2}
\begin{split}
& \iint_{\mathbb{R}^2} ||z|| \exp \Big\{ -\frac {z z^t} 2  \Big\}   z_1^2 z_2^2 d z \\
&= \frac 1 2  \iint_{\mathbb{R}^2} ||z|| \exp \Big\{ -\frac 1 2 z z^t \Big\} (z_1^2 + z_2^2)^2 d z -  \iint_{\mathbb{R}^2} ||z|| \exp \Big\{ -\frac 1 2 z z^t \Big\}   z_1^4 d z 
=  \frac{15}{\sqrt 2} \frac 1 4  \pi^{3/2}.
\end{split}
\end{equation}
Substituting \eqref{eq:(z1^2+z2^2)^2} and \eqref{int z1^2z2^2} into \eqref{14:22}, we obtain
\begin{align*}
\mathcal{I}_3({\bf S}_\ell(\psi))&=\frac{1}{8}  \frac{15}{4 \sqrt 2}  \pi^{3/2}  \Big(  3 S^2_{11;\ell}(\psi)  +  3 S^2_{22;\ell}(\psi)   +2S_{11;\ell}(\psi) S_{22;\ell}(\psi)   \Big)\\
 &= \frac{15 \sqrt 2 }{64} \pi^{3/2}  \{2    {\rm tr} \,{\bf S}^2_{\ell}(\psi) + [  {\rm tr} \,{\bf S}^2_{\ell}(\psi)]^2 \}.
\end{align*}

Write
\begin{align*}
&\mathcal{I}({\bf S}_\ell(\psi))\\
&= \mathcal{I}_0({\bf S}_\ell(\psi)) +\mathcal{I}_1({\bf S}_\ell(\psi))+\mathcal{I}_2({\bf S}_\ell(\psi))+\mathcal{I}_3({\bf S}_\ell(\psi)) +  O({\bf S}^3_{\ell}(\psi))\\
&= \sqrt{2} \pi^{3/2}  +  \frac{3 }{2^{3/2}} \pi^{3/2} \; {\rm tr} \,{\bf S}_{\ell}(\psi)  - \frac{9 }{16 \sqrt 2} \pi^{3/2} \; {\rm tr} \,{\bf S}^2_{\ell}(\psi)+ \frac{15 \sqrt 2 }{64} \pi^{3/2}   [{\rm tr} \,{\bf S}_{\ell}(\psi)]^2   +  O({\bf S}^3_{\ell}(\psi)).
\end{align*}
We finally expand
\begin{align*}
 \frac{1}{ \sqrt{  \text{det} \, {\bf \Delta}_\ell(\psi) }}  = \frac{1}{\sqrt{  \text{det} ( I_2 + {\bf S}_\ell(\psi))}};
\end{align*}
note that
\begin{align*}
\text{det}(I_2+{\bf S}_{\ell}(\psi))&=[1+S_{11;\ell}(\psi)][1+S_{22;\ell}(\psi)]=1+{\rm tr} \,{\bf S}_{\ell}(\psi)  + {\rm det} \,{\bf S}_{\ell}(\psi) ,
\end{align*}
and so,
\begin{align*}
 \frac{1}{ \sqrt{  \text{det} \, {\bf \Delta}_\ell(\psi) }} &=1 - \frac{1}{2} \big[ {\rm tr} \,{\bf S}_{\ell}(\psi) + {\rm det} \,{\bf S}_{\ell}(\psi)  \big] + \frac 3 8  \big[ {\rm tr} \,{\bf S}_{\ell}(\psi) + {\rm det} \,{\bf S}_{\ell}(\psi) \big]^2 + O({\bf S}^3_{\ell}(\psi))\\
&= 1-  \frac{1}{2}  {\rm tr} \,{\bf S}_{\ell}(\psi) -  \frac{1}{2}  {\rm det} \,{\bf S}_{\ell}(\psi) + \frac 3 8  [{\rm tr} \,{\bf S}_{\ell}(\psi)]^2  + O({\bf S}^3_{\ell}(\psi))\\
&=1-  \frac{1}{2}  {\rm tr} \,{\bf S}_{\ell}(\psi) +  \frac{1}{4}  {\rm tr} \,{\bf S}^2_{\ell}(\psi)  + \frac 1 8  [{\rm tr} \,{\bf S}_{\ell}(\psi)]^2  + O({\bf S}^3_{\ell}(\psi)),
\end{align*}
where we have used the fact that  $S^2_{11;\ell}(\psi)$ and $S^2_{22;\ell}(\psi)$ are the eigenvalues of ${\bf S}^2_{\ell}(\psi)$, and we have written ${\rm det} \,{\bf S}_{\ell}(\psi)$ as follows:
\begin{align*}
{\rm det} \,{\bf S}_{\ell}(\psi)
 &= \frac{1}{2} \left\{  [S_{11;\ell}(\psi) + S_{22;\ell}(\psi)]^2 - [S^2_{11;\ell}(\psi)+ S^2_{22;\ell}(\psi)] \right\} = \frac 1 2 \left\{ \left[{\rm tr} \,{\bf S}_{\ell}(\psi) \right]^2 - {\rm tr} \,{\bf S}^2_{\ell}(\psi)  \right\}.
\end{align*}
In conclusion, we have:
\begin{align*}
&K_{1,\ell}(\psi)  \nonumber
 =\frac{\sqrt{\ell(\ell+1)} }{2^2 \pi \sqrt{ \pi}} \Big[   1+ \frac{1}{2} s_\ell(\psi) + \frac 3 8 s^2_\ell(\psi) +O(s^3_\ell(\psi)) \Big]\nonumber  \\
 & \;\; \times  \left[  \sqrt{2} \pi^{3/2}  +  \frac{3 }{2^{3/2}} \pi^{3/2} \; {\rm tr} \,{\bf S}_{\ell}(\psi) - \frac{9 }{16 \sqrt 2} \pi^{3/2} \; {\rm tr} \,{\bf S}^2_{\ell}(\psi)+ \frac{15 \sqrt 2 }{64} \pi^{3/2}   [{\rm tr} \,{\bf S}_{\ell}(\psi)]^2   +  O({\bf S}^3_{\ell}(\psi)) \right] \nonumber \\
 &\;\; \times \left[ 1-  \frac{1}{2}  {\rm tr} \,{\bf S}_{\ell}(\psi) +  \frac{1}{4} {\rm tr} \,{\bf S}^2_{\ell}(\psi) + \frac 1 8  [{\rm tr} \,{\bf S}_{\ell}(\psi)]^2  + O({\bf  S}^3_{\ell}(\psi)) \right] \nonumber \\
&=\frac{ \sqrt{\ell(\ell+1)}}{2^2 \sqrt 2} \left[ 2 + s_{\ell}(\psi) + \frac{1}{2 }{\rm tr} \,{\bf S}_{\ell}(\psi) + \frac{3}{4 } s^2_{\ell}(\psi) + \frac{1}{4 } s_{\ell}(\psi) {\rm tr} \,{\bf S}_{\ell}(\psi) - \frac{1}{16 } {\rm tr} \,{\bf S}^2_{\ell}(\psi)- \frac{1}{32 }    [{\rm tr} \,{\bf S}_{\ell}(\psi)]^2   \right] \nonumber \\
&\;\;+O(\ell \cdot s^3_\ell(\psi)) + O(\ell \cdot {\bf S}^3_\ell(\psi)).
\end{align*}
\end{proof}

\subsection{Proof of Theorem \ref{thm:main asympt}(1)}

\begin{proof}
Substituting the estimates \eqref{sell}, \eqref{Sell11} and \eqref{Sell22} into \eqref{taylorx} we obtain
\begin{align*}
K_{1,\ell}(\psi)&= \frac{\sqrt{\ell(\ell+1)}}{2^2 \sqrt 2} \Big[  2 + 2 \sqrt{\frac 2 \pi} \frac{1}{\sqrt \psi} \cos((\ell+1/2)\psi/\ell-\pi/4) \\
&\;\;+\frac{7}{4 \pi} \frac{1}{\psi}  \cos^2((\ell+1/2)\psi/\ell-\pi/4)-\frac{2}{\pi} \frac{1}{\psi}  \sin^2((\ell+1/2)\psi/\ell-\pi/4) \Big]+O(\psi^{-3/2} \ell^{-2} ),
\end{align*}
and, since $\cos^2(x)=\frac{1}{2}[1+ \cos(2x)]$ and $\sin^2(x)=\frac{1}{2}[1- \cos(2x)]$, we can write
\begin{align*}
&\frac{7}{4 \pi \psi}   \cos^2((\ell+1/2)\psi/\ell-\pi/4)-\frac{2}{\pi \psi}  \sin^2((\ell+1/2)\psi/\ell-\pi/4)\\
&=\frac{7}{4 \pi \psi} \frac 1 2 [1+\cos((\ell+1/2)2\psi/\ell-\pi/2)]-\frac{2}{\pi \psi} \frac 1 2 [1-\cos((\ell+1/2)2\psi/\ell-\pi/2) ]\\
&= \frac{7}{4 \pi \psi} \frac 1 2 -\frac{2}{\pi \psi} \frac 1 2 + \Big[\frac{7}{4 \pi \psi} \frac 1 2 + \frac{2}{\pi \psi} \frac 1 2 \Big] \cos((\ell+1/2)2\psi/\ell-\pi/2)\\
&= -\frac{1}{8 \pi \psi}  + \frac{15}{8 \pi \psi}  \cos((\ell+1/2)2\psi/\ell-\pi/2).
\end{align*}
The above implies
\begin{align*}
K_{1,\ell}(\psi)&= \frac{\sqrt{\ell(\ell+1)}}{2^2 \sqrt 2} \Big[  2 + 2 \sqrt{\frac 2 \pi} \frac{1}{\sqrt \psi} \cos((\ell+1/2)\psi/\ell-\pi/4) \\
&\;\;-\frac{1}{8 \pi \psi}  + \frac{15}{8 \pi \psi}  \cos((\ell+1/2)2\psi/\ell-\pi/2)\Big]+O(\psi^{-3/2} \ell^{-2} )\\
&= \frac{\sqrt{\ell(\ell+1)}}{2 \sqrt 2} \Big[  1 +  \sqrt{\frac 2 \pi} \frac{1}{\sqrt \psi} \cos((\ell+1/2)\psi/\ell-\pi/4) \\
&\;\;-\frac{1}{16 \pi \psi}  + \frac{15}{16 \pi \psi}  \cos((\ell+1/2)2\psi/\ell-\pi/2)\Big]+O(\psi^{-3/2} \ell^{-2} ),
\end{align*}
the statement \eqref{eq:nod bias hemi far from boundary} of Theorem \ref{thm:main asympt}(1).

\end{proof}

\section{Proof of Theorem \ref{thm:main asympt}(2): perturbative analysis at the boundary}
The aim of this section is to study the asymptotic behaviour of the density function $K_{1,\ell}(\psi)$ for $0 < \psi < \epsilon_0$ with $\epsilon_0>0$ sufficiently small. We have
\begin{align*}
K_{1,\ell}(\psi)= \frac{1}{\sqrt{2 \pi} \sqrt{1-P_{\ell}(\cos(\psi/\ell))}}  \frac{1}{2 \pi  \sqrt{  \text{det}\, {{\bf \Delta}}_\ell(\psi)}}
\sqrt{\ell (\ell+1)} \iint_{\mathbb{R}^2} ||z|| \exp \Big\{ -\frac 1 2 z^t {{\bf \Delta}}^{-1}_\ell(\psi) z \Big\} d z,
\end{align*}
where ${{\bf \Delta}}_\ell(\psi)$ is the scaled conditional covariance matrix
\begin{align*}
{{\bf \Delta}}_\ell(\psi)&=  {{\bf C}}_{\ell}(\psi)-  \frac{{{\bf B}}^t_{\ell}(\psi) {{\bf B}}_{\ell}(\psi)}{1-P_{\ell}(\cos(\psi/\ell ))}.
\end{align*}
We have that
\begin{align} \label{omp}
1-P_{\ell}(\cos(\psi/\ell ))&=\frac{\ell (\ell+1)}{\ell^2} \frac{\psi^2}{2^2} - \frac{(\ell-1) \ell (\ell+1) (\ell+2)}{4\, \ell^4} \frac{\psi^4}{2^4}  \\
&\;\;+ \frac{1}{36} \frac{(\ell-2)(\ell-1) \ell (\ell+1) (\ell+2)(\ell+3)}{\ell^6} \frac{\psi^6}{2^6}+ O(\psi^8), \nonumber 
\end{align}
with constant involved in the $`O'$-notation absolute. We also have
\begin{align*}
{{\bf B}}^t_{\ell}(\psi)&=\left( \begin{array}{c}
 -\sin (\psi/\ell ) P'_{\ell}(\cos(\psi/\ell )) \left( - \frac{1}{\ell} \right)\\ 0
 \end{array}\right)\\
 &=  \left( \begin{array}{c}
\frac{\ell (\ell+1)}{\ell^2}  \frac{\psi}{2} -  \frac{(\ell-1) \ell (\ell+1) (\ell+2)}{2 \, \ell^4} \frac{\psi^3}{2^3}+\frac{1}{12} \frac{(\ell-2)(\ell-1) \ell (\ell+1) (\ell+2)(\ell+3)}{l^6} \frac{\psi^5}{2^5} +O(\psi^7) \\ 0
 \end{array}\right),
\end{align*}
and ${{\bf C}}_{\ell}(\psi)$ is the $2 \times 2$ symmetric matrix with entries
\begin{align*}
{C}_{\ell,11}(\psi)&=\left[ P'_{\ell}(1)+ \cos(\psi/\ell ) \; P'_{\ell}(\cos(\psi/\ell )) -  \sin^2(\psi/\ell ) \; P''_{\ell}(\cos(\psi/\ell )) \right] \left( - \frac{1}{\ell} \right)^2\\
&=1-\frac{3}{4} \frac{(\ell-1) \ell (\ell+1) (\ell+2)}{ \ell^4} \frac{\psi^2}{2^2} + \frac{5}{24} \frac{(\ell-2)(\ell-1) \ell (\ell+1) (\ell+2)(\ell+3)}{\ell^6 } \frac{\psi^4}{2^4} + O(\psi^6),\\
{{\bf C}}_{\ell,12}(\psi)&= 0,\\
{{\bf C}}_{\ell,22}(\psi)& = \left[ P'_{\ell}(1)-P'_{\ell}(\cos(\psi/\ell )) \right] \left( - \frac{1}{\ell} \right)^2\\
&=\frac{(\ell-1) \ell (\ell+1) (\ell+2)}{4\, \ell^4} \frac{\psi^2}{2^2} - \frac{(\ell-2)(\ell-1) \ell (\ell+1) (\ell+2)(\ell+3)}{24 \ell^6 } \frac{\psi^4}{2^4} + O(\psi^6).
 \end{align*}
We obtain that
\begin{align*}
{{\bf \Delta}}_\ell(\psi)&=\left( \begin{array}{cc}
{\delta}_{11,\ell}(\psi)& 0\\
0& {\delta}_{22,\ell}(\psi)
\end{array}\right),
\end{align*}
with
\begin{align} \label{19:30}
{\delta}_{11,\ell}(\psi)&=\frac{1}{2^8 3^2} \frac{(\ell-2)(\ell-1) \ell (\ell+1) (\ell+2)(\ell+3)}{\ell^6} \psi^4+ O(\psi^6) \nonumber \\
&= \frac{1}{2^8 3^2}  \psi^4+ O( \ell^{-1} \psi^4)+O(\psi^6),
\end{align}
and
\begin{align} \label{eq:19:31}
{\delta}_{22,\ell}(\psi)&= \frac{(\ell-1) \ell (\ell+1) (\ell+2)}{4\, \ell^4} \frac{\psi^2}{2^2}+ O(\psi^4)= \frac{\psi^2}{16} + O( \ell^{-1} \psi^2)+ O(\psi^4).
\end{align}
We introduce the change of variable $\xi={{\bf \Delta}}^{-1/2}_\ell(\psi) z$, and we write
\begin{align*}
K_{1,\ell}(\psi)= \frac{1}{\sqrt{2 \pi} \sqrt{1-P_{\ell}(\cos(\psi/\ell))}}  \frac{1}{2 \pi }
\sqrt{\ell (\ell+1)} \iint_{\mathbb{R}^2} \sqrt{{\delta}_{11,\ell}(\psi) \xi_1^2+{\delta}_{22,\ell}(\psi) \xi_2^2} \exp \Big\{ -\frac{\xi^t  \xi}2  \Big\} d \xi.
\end{align*}
Using the expansions in \eqref{omp}, \eqref{19:30} and \eqref{eq:19:31}, we write
\begin{align*}
K_{1,\ell}(\psi)&= \frac{1}{\sqrt{2 \pi} \sqrt{\psi^2/4 + O(\ell^{-1}\psi^2) + O(\psi^4)} }
\sqrt{\ell (\ell+1)}  \left[ \frac{\psi}{4} + O(\ell^{-1}\psi) + O(\psi^3) \right] \sqrt{\frac{2}{\pi}}\\
&= \sqrt{\ell (\ell+1)}  \frac{1}{2\pi} + O(1) + O(\ell \psi^2),
\end{align*}
which is \eqref{eq:nod bias hemi close to boundary}.

\section{Proof of Corollary \ref{cor}: expected nodal length}

\subsection{Kac-Rice formula for expected nodal length}

The Kac-Rice formula is a meta-theorem allowing one to evaluate the moments of the zero set of a random field satisfying some smoothness and non-degeneracy conditions. For $F:\R^{d}\rightarrow\R$, a sufficiently smooth centred Gaussian random field, we define
\begin{equation*}
K_{1,F}(x):= \frac{1}{\sqrt{2\pi} \sqrt{\var(F(x))}}\cdot \E[|\nabla F(x)|\big| F(x)=0]
\end{equation*}
the zero density (first intensity) of $F$. Then the Kac-Rice formula asserts that for some suitable class of random fields $F$
and $\overline{\Dpc}\subseteq \R^{d}$ a compact closed subdomain of $\R^{d}$, one has the equality
\begin{equation}
\label{eq:Kac Rice meta}
\E[\vol_{d-1}(F^{-1}(0)\cap \overline{\Dpc})] = \int\limits_{\overline{\Dpc}}K_{1,F}(x)dx.
\end{equation}

We would like to apply \eqref{eq:Kac Rice meta} to the boundary-adapted random spherical harmonics $T_{\ell}$ to evaluate the asymptotic law of the total expected nodal length of $T_{\ell}$. Unfortunately the aforementioned non-degeneracy conditions fail at the equator $$\mathcal{E}=\{(\theta, \phi):\:  \theta=\pi/2 \}\subseteq\Hc^{2}.$$ Nevertheless, in a manner inspired by ~\cite[Proposition 2.1]{CKW}, we excise a small neighbourhood of this degenerate set, and apply the Monotone Convergence Theorem so to be able to prove that \eqref{eq:Kac Rice meta} holds precisely, save for the length of the equator that is bound to be contained in the nodal set of $T_{\ell}$, by the Dirichlet boundary condition.

\begin{proposition} \label{KKRR}
The expected nodal length of $T_{\ell}$ satisfies
\begin{equation} \label{cs}
\E[\Lc({T_{\ell}})] = \int_{\Hc^2} K_{1,\ell} (x) d x + 2 \pi,
\end{equation}
where $K_{1,\ell} (\cdot)$ is the zero density of $T_{\ell}$.
\end{proposition}
\begin{proof}
One way justify the Kac-Rice formula outside the equator is by using \cite[Theorem 6.8]{AW}, that assumes the non-degeneracy of the $3\times 3$ covariance matrix at all these points, a condition we were able to verify via an explicit, though somewhat long,
computation, omitted here. Alternatively,
to validate the Kac-Rice formula it is sufficient \cite[Lemma 3.7]{KW}
that the Gaussian distribution of $T_{\ell}$ is non-degenerate for every $x \in \Hc^2\setminus \mathcal{E}$, which is easily satisfied.

We construct a small neighbour of the equator $\mathcal{E}$, i.e. the set
$$\mathcal{E}_{\varepsilon}=\left\{(\theta, \phi):\: \theta \in \left [\frac \pi 2 , \frac \pi 2-  \varepsilon \right ) \right\},$$
and we denote
$$ \Hc_{\varepsilon}= \Hc \setminus  \mathcal{E}_{\varepsilon}.$$
Since Kac-Rice formula holds for $T_{\ell}$ restricted to $\Hc_{\varepsilon}$, the expected nodal length for $T_{\ell}$ restricted to $\Hc_{\varepsilon}$ is
\begin{equation*}
\E[\Lc({T_{\ell}}|_{\Hc_{\varepsilon}})] = \int_{\Hc_{\varepsilon}} K_{1,\ell} (x) d x.
\end{equation*}
Since the restricted nodal length $\{ \Lc({T_{\ell}}|_{\Hc_{\varepsilon}}) \}_{\varepsilon>0}$ is an increasing sequence of nonnegative random variables with a.s. limit
$$\lim_{\varepsilon \to 0} \Lc({T_{\ell}}|_{\Hc_{\varepsilon}}) = \Lc({T_{\ell}}) - 2 \pi,$$
the Monotone Convergence Theorem yields
\begin{equation} \label{lim1}
\lim_{\varepsilon \to 0} \E[\Lc({T_{\ell}}|_{\Hc_{\varepsilon}})] = \E[\Lc({T_{\ell}}) ]- 2 \pi.
\end{equation}
Moreover, by the definition
\begin{equation} \label{lim2}
\lim_{\varepsilon \to 0}  \int_{\Hc_{\varepsilon}} K_{1,\ell} (x) d x =  \int_{\Hc} K_{1,\ell} (x) d x.
\end{equation}
The equality of the limits in \eqref{lim1} and \eqref{lim2} show that Proposition \ref{KKRR} holds.
\end{proof}

\subsection{Expected nodal length}

\begin{proof}[Proof of Corollary \ref{cor}]

To analyse asymptotic behaviour of the expected nodal length, we separate the contribution of the following three subregions of the hemisphere $\Hc$ in the Kac-Rice integral on the r.h.s of \eqref{cs}:
$$ \Hc_C= \{ (\psi, \phi): 0 < \psi < \epsilon_0\}, \hspace{0.5cm} \Hc_I= \{ (\psi, \phi): \epsilon_0 < \psi < C \}, \hspace{0.5cm} \Hc_F= \{ (\psi, \phi): C < \psi < \pi \ell\};$$
note that we express the three subregions of $\Hc$ in terms of the scaled variable $\psi$. In what follows we argue that $\Hc_{F}$ gives
the main contribution. \\

In the (scaled) spherical coordinates we may rewrite the Kac-Rice integral \eqref{cs} as
\begin{equation*}
\E[\Lc({T_{\ell}})]-2\pi=\frac{ \pi}{ \ell}  \int_{0}^{\ell \pi } K_{1,\ell}(\psi) \sin\left(\frac \pi 2 - \frac{\psi}{2 \ell} \right)   d \psi,
\end{equation*}
and the contribution of the third range $\Hc_{F}$ as
\begin{align}
\label{eq:E[] int far}
\E[\Lc({T_{\ell}}|_{\Hc_{F}})] & =   \frac{ \pi}{ \ell}  \int_{C}^{\ell \pi } K_{1,\ell}(\psi) \sin\left(\frac \pi 2 - \frac{\psi}{2 \ell} \right)   d \psi.
\end{align}
We are now going to invoke the asymptotics of $K_{1,\ell}(\psi)$, prescribed by \eqref{eq:nod bias hemi far from boundary} for this range.
The first term in \eqref{eq:nod bias hemi far from boundary} contributes
\begin{equation}
\label{eq:1st term contr}
\begin{split}
\frac{\pi}{ \ell}  \frac{\sqrt{\ell(\ell+1)}}{2 \sqrt 2} \int_{C}^{\ell \pi }  \sin\left(\frac \pi 2 - \frac{\psi}{2 \ell} \right)   d \psi &= \frac{\pi}{ \ell}  \frac{\sqrt{\ell(\ell+1)}}{2 \sqrt 2} 2 \ell \left[1-\sin\left( \frac{C}{2 \ell}\right) \right] \\
&=    2 \pi   \frac{\sqrt{\ell(\ell+1)}}{2 \sqrt 2}  \left[1- \frac{C}{2 \ell} + O\left( \frac{C}{\ell}\right) \right].
\end{split}
\end{equation}
to the integral \eqref{eq:E[] int far}.
The second term in \eqref{eq:nod bias hemi far from boundary} gives
\begin{equation}
\label{eq:2nd term cont}
\begin{split}
& \frac{\pi}{ \ell} \frac{\sqrt{\ell(\ell+1)}}{2 \sqrt 2} \int_{C}^{\ell \pi } \sqrt{\frac 2 \pi} \frac{1}{\sqrt \psi} \cos\{(\ell+1/2)\psi/\ell-\pi/4\}   \sin\left(\frac \pi 2 - \frac{\psi}{2 \ell} \right)   d \psi = O(\ell^{-1/2}),
\end{split}
\end{equation}
since, upon transforming the variables $w= \psi/\ell$, this term is bounded by
\begin{align*}
& \sqrt{\ell}  \int_{C/\ell}^{\pi}    \frac{1}{\sqrt{w} } \cos\{(\ell+1/2) w -\pi/4\}   d w \\
& = \frac{\sqrt{\ell}}{\sqrt 2}  \int_{C/\ell}^{\pi}  \frac{1}{\sqrt{w} } [ \cos\{(\ell+1/2) w\}  + \sin \{(\ell+1/2) w\} ]  d w \\
&=\frac{\sqrt{\ell}}{\sqrt 2}  \left\{  \left.  \frac{1}{\sqrt{w} } \frac{\sin((\ell+1/2) w)}{\ell+1/2} \right|_{C/\ell}^{\pi} + \frac{1}{2}  \int_{2 a_{\ell}/\ell}^{\pi}  w^{-3/2}  \frac{\sin((\ell+1/2) w)}{\ell+1/2} d w \right\} \\
&+ \frac{\sqrt{\ell}}{\sqrt 2}  \left\{  \left. -  \frac{1}{\sqrt{w} } \frac{\cos((\ell+1/2) w)}{\ell+1/2} \right|_{C/\ell}^{\pi} - \frac{1}{2}  \int_{2 a_{\ell}/\ell}^{\pi}  w^{-3/2}  \frac{\cos((\ell+1/2) w)}{\ell+1/2} d w \right\} \\
&= O(1/\sqrt{\ell}).
\end{align*}
The logarithmic bias is an outcome of
\begin{equation}
\label{eq:log bias}
\begin{split}
  \frac{\pi}{ \ell}  \frac{\sqrt{\ell(\ell+1)}}{2 \sqrt 2} \int_{C}^{\ell \pi } \left( -\frac{1}{16 \pi \psi}  \right)  \sin\left(\frac \pi 2 - \frac{\psi}{2 \ell} \right)   d \psi &= - \frac{1}{16  \ell}  \frac{\sqrt{\ell(\ell+1)}}{2 \sqrt 2} \left[ - \log \left(\frac{C}{2 \ell} \right) +O(1) \right]  \\
&=-  \frac{1}{16  \ell}  \frac{\sqrt{\ell(\ell+1)}}{2 \sqrt 2}  \log(\ell) +O( 1).
\end{split}
\end{equation}
Consolidating all of the above estimates \eqref{eq:1st term contr},
\eqref{eq:2nd term cont} and \eqref{eq:log bias}, and the contribution of the error term
in \eqref{eq:nod bias hemi far from boundary}, we finally obtain
\begin{align*}
\E[\Lc({T_{\ell}}|_{\Hc_{F}})] =   2 \pi   \frac{\sqrt{\ell(\ell+1)}}{2 \sqrt 2} -  \frac{1}{16  \ell}  \frac{\sqrt{\ell(\ell+1)}}{2 \sqrt 2}  \log(\ell) + O(1).
\end{align*}

\vspace{0.3cm}

The contribution to the Kac-Rice integral on the r.h.s of \eqref{cs} of the set $\Hc_{C}$ is bounded by the straightforward
\begin{align*}
\E[\Lc({T_{\ell}}|_{\Hc_{C}})] &= \frac{ \pi}{ \ell}  \int_{0}^{\varepsilon_0 } K_{1,\ell}(\psi) \sin\left(\frac \pi 2 - \frac{\psi}{2 \ell} \right)   d \psi=O(1),
\end{align*}
on recalling the uniform estimate \eqref{eq:nod bias hemi close to boundary}.
Finally, we may bound the contribution of the intermediate range $\Hc_I$ as follows. We first write
\begin{align*}
\E[\Lc({T_{\ell}}|_{\Hc_{I}})] &=\frac{1}{\sqrt{2\pi}}  \int_{\Hc_I}  \frac{1}{\sqrt{1-P_{\ell}(\cos(\frac{\psi}{\ell}))}}\cdot \E\left[\left \|\nabla T_{\ell}\left({\psi}/{\ell}\right) \right\|\big| T_{\ell}({\psi}/{\ell})=0\right] d \psi
\end{align*}
then we observe that on the intermediate range
$$\Hc_I=\left\{({\psi}/{\ell}, \phi): \varepsilon_0< \psi< C \right\},$$
the variance at the denominator, i.e. $1-P_{\ell}(\cos(\psi/\ell))$, is bounded away from $0$, and moreover the diagonal entries of the unconditional covariance matrix ${\bf C}_{\ell}$ of the Gaussian vector $\nabla T_{\ell}$ are $O(\ell^2)$, and so are the diagonal entries of the conditional matrix ${\bf \Omega}_\ell$, since they are bounded by the unconditional ones, as it follows directly from \eqref{eq:Omega transition}, or, alternatively, from the vastly general Gaussian Correlation Inequality ~\cite{Royen}.
This easily gives the following upper bound:
$$\mathbb{E}[\|\nabla T_{\ell}({\psi}/{\ell}) \| \big|  T_{\ell}({\psi}/{\ell})=0] \le \left(\mathbb{E}[\|\nabla T_{\ell}({\psi}/{\ell}) \|^2 \big| T_{\ell}({\psi}/{\ell})=0]\right)^{1/2} \le \left(\mathbb{E}[\| \nabla T_{\ell}({\psi}/{\ell}) \|^2 ] \right)^{1/2} =O(\ell) .$$
Since the area of $\Hc_C$ is $O(\ell^{-1})$, it follows that the total contribution this range to the expected nodal length is $O(1)$.

\end{proof}

\appendix

\section{Proof of Proposition \ref{16:52}}
We have that
\begin{align*}
\E[T_{\ell}(x)\cdot T_{\ell}(y)] &=\frac{8 \pi}{2 \ell+1} \sum\limits_{\substack{m=-\ell\\m\not\equiv\ell\mod{2}}}^{\ell} Y_{\ell,m} (x)  \;  \overline{Y}_{\ell, m} (y)\\
&=\frac 1 2 \frac{8 \pi}{2 \ell+1}\left[  \sum_{m=-\ell}^{\ell} Y_{\ell,m} (x)  \;  \overline{Y}_{\ell, m} (y)+\sum_{m=-\ell}^{\ell} (-1)^{m+\ell+1} Y_{\ell,m} (x)  \;  \overline{Y}_{\ell, m} (y) \right]\\
&=\frac 1 2 \frac{8 \pi}{2 \ell+1}\left[  \sum_{m=-\ell}^{\ell} Y_{\ell,m} (x)  \;  \overline{Y}_{\ell, m} (y) -\sum_{m=-\ell}^{\ell}  Y_{\ell,m} (\overline{x})  \;  \overline{Y}_{\ell,m} (y) \right],
\end{align*}
where we have used the fact that $Y_{\ell,m} (\theta, \phi)=(-1)^{\ell+m} Y_{\ell,m} (\pi- \theta, \phi)$.
We apply now the Addition Theorem for Spherical Harmonics:
\begin{align*}
P_{\ell}(\cos d(x,y))= \frac{4 \pi}{2 \ell+1} \sum_{m=-\ell}^{\ell} Y_{\ell,m} (x)  \;  \overline{Y}_{\ell,m} (y),
\end{align*}
so that
\begin{align*}
\E[T_{\ell}(x)\cdot T_{\ell}(y)] =P_{\ell}(\cos d(x,y)) - P_{\ell}(\cos d(\overline{x},y)).
\end{align*}
\begin{remark}
In particular, we note that,
\begin{align*}
\mathbb{E}[ T^2_{\ell}(x)] =  P_{\ell}( \langle x,x \rangle) - P_{\ell}(\langle \overline{x} ,x  \rangle) = 1- P_{\ell}(\cos(\pi-2 \theta)),
\end{align*}
this implies
\begin{align*}
\text{Var}(T_{\ell}(x))&=
\begin{cases}
1-P_{\ell}(\cos(\pi))= 1-(-1)^{\ell}&  \rm{if\;\;} \theta=0,\\
1- P_{\ell}(1)=0&  \rm{if\;\;} \theta=\pi/2,\\
\to 1 \text{ as } \ell \to \infty &  \rm{if\;\;} \theta\ne 0, \pi/2.
\end{cases}
\end{align*}
Moreover, as $\ell \to \infty$, for $\theta\ne 0, \pi/2$,
\begin{align*}
\frac{\E[T_{\ell}(x)\cdot T_{\ell}(y)] }{P_{\ell}(\cos d(x,y))} \to 1.
\end{align*}
\end{remark}

\end{document}